\documentclass[11pt]{llncs}
\usepackage{RR}
\RRNo{6186}
\usepackage[english]{babel}
\usepackage{latexsym}
\usepackage{amsfonts,amssymb}
\usepackage{comment}
\usepackage{latexsym}
\usepackage{a4wide}
\usepackage{url}
\usepackage[only,llbracket,rrbracket]{stmaryrd}
\usepackage[latin1]{inputenc}
\usepackage{aeguill}
\usepackage[dvips]{epsfig}
\usepackage{times}

\def\eps{\varepsilon}
\def\qed{\hfill$\Box$}
\def\fxu{\lfloor x_1 \rfloor}
\def\fxd{\lfloor x_2 \rfloor}
\def\lu{\lambda_1}
\def\ld{\lambda_2}
\def\pis{\tilde{\pi}}
\newcommand{\iinterv}[1]{\left\llbracket#1\right\rrbracket}
\newcommand{\iintervo}[1]{\left\llbracket#1\right)}
\newcommand{\assign}{:=}

\newcommand {\aplt} {\ {\raise-.5ex\hbox{$\buildrel<\over\sim$}}\ }

\def\E{{\mathcal E}}
\def\Z{{\mathbb Z}}
\def\R{{\mathbb R}}

\newcommand{\card}[1]{\left|#1\right|}

\newcounter{algorithmeligne}
\newdimen\framewidth

\def\myframe#1{%
  \framewidth=\textwidth
  \advance\framewidth by -2pt   
  \advance\framewidth by -\rightmargin
  \advance\framewidth by -\leftmargin
  \fbox{\begin{minipage}{\framewidth}#1\end{minipage}}%
}

\def\algorithme#1
{
       \setcounter{algorithmeligne}{0}
        \newcommand{\ITEM}{
                \stepcounter{algorithmeligne}
                {\footnotesize \arabic{algorithmeligne}.}
                }
        \newcommand{\lign}{\ITEM}
  \framewidth=\textwidth
  \advance\framewidth by -1pt   
  \advance\framewidth by -\rightmargin
  \advance\framewidth by -\leftmargin
\begin{center}   
  \fbox{\hspace{0.08cm}\begin{minipage}{\framewidth} 
        \small
        \,
	\vspace*{-.3cm}
        \begin{tabbing}
\hskip0.5cm\=\qquad\=\qquad\=\qquad\=\qquad\=\qquad\=\qquad\=\qquad\=\qquad\=\qquad\=\qquad\\\kill 
        #1
        \end{tabbing}
\end{minipage}\hspace{0.1cm}}%
        \end{center} 
\vspace*{-.4cm}}

\def\mR{{\mathbb R}}  
\def\mZ{{\mathbb Z}}
\newcommand{\ps}[2]{\langle#1,#2\rangle}
\def\eps{\epsilon}

\def\assign{{:=}}

\RRdate{Mai 2007}
\RRetitle{Improved Analysis of Kannan's Shortest Lattice Vector Algorithm}
\RRtitle{Am\'elioration de l'analyse de l'algorithme de Kannan pour le 
probl\`eme du vecteur le plus court}
\RRauthor{Guillaume Hanrot\and Damien Stehl\'e\thanks{
CNRS and \'{E}cole Normale Sup\'erieure de Lyon, LIP, 46 all\'ee d'Italie,
69007 Lyon, France.}}
\RRprojets{Ar\'enaire et Cacao}
\RRtheme{\THSym}
\URLorraine
\RRkeyword{
Lattice reduction, complexity analysis, lattice-based cryptosystems.
}
\RRmotcle{R\'eduction des r\'eseaux, analyse de complexit\'e, 
cryptosyst\`emes bas\'es sur les r\'eseaux}

\def\Z{{\mathbb Z}}

\RRabstract{
The security of lattice-based cryptosystems such as NTRU, GGH and
Ajtai-Dwork essentially relies upon the intractability of computing a
shortest non-zero lattice vector and a closest lattice vector to a
given target vector in high dimensions.  The best algorithms for these
tasks are due to Kannan, and, though remarkably simple, their
complexity estimates have not been improved since more than twenty
years. Kannan's algorithm for solving the shortest vector problem is
in particular crucial in Schnorr's celebrated block reduction
algorithm, on which are based the best known attacks against the
lattice-based encryption schemes mentioned above. Understanding
precisely Kannan's algorithm is of prime importance for providing
meaningful key-sizes. In this paper we improve the complexity analyses
of Kannan's algorithms and discuss the possibility of improving
the underlying enumeration strategy.
}

\RRresume{
La s\'ecurit\'e des cryptosyst\`emes bas\'es sur les r\'eseaux, tels
NTRU, GGH, ou encore Ajtai-Dwork, repose essentiellement sur la difficult\'e
\`a calculer un vecteur non nul le plus court, ou le plus proche d'un
vecteur cible donn\'e, en grande dimension. Les meilleurs algorithmes pour
accomplir ces t\^aches sont dus \`a Kannan, et, en d\'epit de leur 
grande simplicit\'e, l'analyse de leur complexit\'e n'a pas
\'et\'e am\'elior\'ee depuis plus de 20 ans. L'algorithme de Kannan pour
r\'esoudre le probl\`eme du vecteur le plus court est particuli\`erement
critique dans le c\'el\`ebre algorithme de Schnorr pour la r\'eduction 
par blocs, sur lequel sont bas\'ees les meilleures attaques contre les
sch\'emas de chiffrement utilisant les r\'eseaux mentionn\'ees 
pr\'ec\'edemment. Comprendre pr\'ecis\'ement la complexit\'e de l'algorithme
de Kannan est donc crucial pour d\'eterminer des tailles de cl\'e pertinentes.
Dans ce travail, nous am\'eliorons les analyses de complexit\'e des
algorithmes de Kannan, et discutons la possibilit\'e d'am\'eliorer la 
strat\'egie d'\'enum\'eration sous-jacente. 
}
\begin{document}
\makeRR

\section{Introduction}
A lattice~$L$ is a discrete subgroup of some~$\mR^n$. Such an object
can always be represented as the set of integer linear combinations of
no more than~$n$ vectors~$\vec{b}_1, \ldots, \vec{b}_d$. If these
vectors are linearly independent, we say that they are a basis of the
lattice~$L$. The most famous algorithmic problem associated with
lattices is the so-called Shortest Vector Problem (SVP). Its
computational variant is to find a non-zero lattice vector of smallest
Euclidean length~---~this length being the minimum~$\lambda(L)$ of the
lattice~---~given a basis of the lattice. Its decisional variant is
known to be NP-hard under randomised reductions~\cite{Ajtai98}, even
if one only asks for a vector whose length is no more than~$2^{(\log
d)^{1-\eps}}$ times the length of a shortest vector~\cite{HaRe06} (for
any~$\eps>0$).

\medskip

SVP is of prime importance in cryptography since a now quite large
family of public-key cryptosystems rely more or less on it. The
Ajtai-Dwork cryptosystem~\cite{AjDw97} relies on $d^c$-SVP for
some~$c>0$, where~$f(d)$-SVP is the problem of finding the shortest
non-zero vector in the lattice~$L$, knowing that it is unique in the
sense that any vector that is of
length less than~$f(d) \cdot \lambda(L)$ is parallel to it.  The GGH
cryptosystem~\cite{GoGoHa97} relies on special instances of the Closest
Vector Problem (CVP), a non-homogeneous version of SVP. 
Finally, one strongly suspects that in
NTRU~\cite{HoPiSi98}~--~the only realistic lattice-based cryptosystem
nowadays, the private key can be read on the coordinates
of a shortest vector of the Coppersmith-Shamir lattice~\cite{CoSh97}.
The best known generic attacks on these encryption schemes are based
on solving SVP. It is therefore highly
important to know precisely what complexity is achievable, both in
theory and practice, in particular to select meaningful key-sizes.

\medskip

In practice, when one wants to obtain good approximations of the
lattice minimum, one uses Schnorr's block-based
algorithms~\cite{Schnorr87,ScEu94}.  These algorithms use internally
either Kannan's algorithm, or the lattice point enumeration procedure
on which it relies.  This is by far the most time-consuming part of
these algorithms.  In fact, the corresponding routine in Shoup's
NTL~\cite{NTL} relies on a much slower algorithm described
in~\cite{ScEu94} ($2^{O(d^2)}$ instead of~$d^{O(d)}$). The problem is
that the enumeration is performed on a basis which is not sufficiently
pre-processed (only LLL-reduced).  It works well in low dimension, but
it can be checked that it is sub-optimal even in moderate dimensions
(say~40): the efficiency gap between enumerating from
an LLL-reduced basis and from an HKZ-reduced basis shows that there is much
room for improving the strategy of~\cite{ScEu94} by pre-processing
the basis before starting the enumeration.

\medskip

Two main algorithms are known for solving SVP.  The first one, which
is deterministic, is based on the exhaustive enumeration of lattice
points within a small convex set. It is known as Fincke-Pohst's
enumeration algorithm~\cite{FiPo83} in the algorithmic number theory
community. In the cryptography community, it is known as Kannan's
algorithm~\cite{Kannan83}, which is quite similar to the one of Fincke
and Pohst. There are two main differences between both: firstly, in
Kannan's algorithm, a long pre-computation on the basis is performed
before starting the enumeration process; secondly, Kannan enumerates
points in a hyper-parallelepiped whereas Fincke and Pohst do it in an
hyper-ellipsoid contained in Kannan's hyper-parallelepiped~--~though it
may be that Kannan chose the hyper-parallelepiped in order to
simplify the complexity analysis.  Kannan obtained a~$d^{d+o(d)}$ complexity
bound (in all the complexity bounds mentioned in the introduction,
there is an implicit multiplicative factor that is polynomial in the
bit-size of the input). In 1985, Helfrich~\cite{Helfrich85} refined
Kannan's analysis, and obtained a~$d^{d/2+o(d)}$ complexity bound. On
the other hand, Ajtai, Kumar and Sivakumar~\cite{AjKuSi01} described a
probabilistic algorithm of complexity~$2^{O(d)}$. The best exponent
constant is likely to be small.  Nevertheless, unless a breakthrough
modification is introduced, this algorithm is bound to remain
impractical even in moderate dimension since it also requires an
exponential space (at least~$2^d$ in dimension~$d$).  
On the contrary, the deterministic algorithm of
Kannan requires a polynomial space.

\medskip

Our main result is to lower Helfrich's complexity bound on Kannan's
algorithm, from~$d^{\frac{d}{2}+o(d)} \approx d^{0.5 \cdot d}$
to~$d^{\frac{d}{2e}+o(d)} \approx d^{0.184 \cdot d+o(d)}$. This may
explain why Kannan's algorithm is tractable even in moderate
dimensions (higher than~$40$). Our
analysis can also be adapted to Kannan's algorithm that solves the
Closest Vector Problem: it decreases Helfrich's complexity
bound from~$d^{d+o(d)}$ to~$d^{d/2+o(d)}$. The complexity improvement
on Kannan's SVP algorithm directly provides better worst-case 
efficiency/quality trade-offs in Schnorr's block-based
algorithms~\cite{Schnorr87,ScEu94,GaHoKoNg06}.

It must be noted that if one
follows our analysis step by step, the derived~$o(d)$ may be large
when evaluated for some practical~$d$: the constants hidden in the
``$o(d)$'' are improvable (for some of them it may be easy, for
others it is probably much harder). No effort was made to improve
them, and we believe that it would have complicated the
proof with irrelevant details. In fact, most of our analysis consists
of estimating the number of lattice points within convex bodies, and showing
that the approximation by the volume is valid. By replacing this 
discretization by heuristic volume estimates, one obtains very small
heuristic hidden constants.

\medskip

Our complexity improvement is based on a fairly simple idea.  It is
equivalent to generate all lattice points within a ball and to
generate all integer points within an ellipsoid (consider the
ellipsoid defined by the quadratic form naturally associated with the
given lattice basis). Fincke and Pohst noticed that it was more
efficient to work with the ellipsoid than to consider a parallelepiped
containing it: indeed, when the dimension increases, the ratio of the two
volumes shrinks to~$0$ very quickly.
Amazingly, in his analysis, instead of considering the ellipsoid, 
Kannan bounds the volume of the parallelepiped. Using rather
involved technicalities, we bound the volume of the ellipsoid
(in fact, the number of integer points within it).
Some parts of our proof could be of independent interest. For example,
we show that for any Hermite-Korkine-Zolotarev-reduced (HKZ-reduced for short)
lattice basis~$(\vec{b}_1, \ldots, \vec{b}_d)$, and any subset~$I$
of~$\{1,\ldots,d\}$, we have:
$$ \frac{\|\vec{b}_1\|^{\card{I}}}{\prod_{i \in I} \|\vec{b}_i^*\|} \leq 
\sqrt{d}^{\card{I} \left(1 + \log \frac{d}{\card{I}}\right)},$$ 
where~$(\vec{b}_i^*)_{i \leq d}$ is the Gram-Schmidt orthogonalisation
of the basis~$(\vec{b}_1, \ldots, \vec{b}_d)$. This inequality generalises 
the results of~\cite{Schnorr87} on the quality of HKZ-reduced bases.

\medskip
\noindent {\sc Road-Map of the Paper.} In Section~\ref{se:background},
we recall some basic definitions and properties on lattice reduction.
Section~\ref{se:kannan} is devoted to the description of Kannan's
algorithm and Section~\ref{se:ana} to its complexity analysis.  In
Section~\ref{se:cvp}, we give without much detail our sibling result
on CVP, as well as very direct consequences of our result for
Schnorr's block-based algorithms. 

\medskip
\noindent {\sc Notation.} All logarithms are natural logarithms, i.e.,
$\log(e) = 1$.  Let~$\| \cdot \|$ and~$\ps{ \cdot }{ \cdot }$ be the
Euclidean norm and inner product of~$\mR^n$. Bold variables are
vectors. We use the bit complexity model. The 
notation~$\mathcal{P}(n_1,\ldots, n_i)$ 
means~$(n_1 \cdot \ldots \cdot n_i)^c$ 
for some constant~$c>0$. If~$x$ is real, we denote 
by~$\lfloor x \rceil$ a closest integer to it (with any convention for
making it unique) and we define the centred fractional part~$\{x\}$ 
as~$x-\lfloor x \rceil$. We use the notation~$\mbox{frac}(x)$ to denote 
the classical fractional part of~$x$, i.e., 
the quantity~$x - \lfloor x \rfloor$. 
Finally, for any integers~$a$ and~$b$, we
define~$\iinterv{a,b}$ as~$[a,b] \cap \mZ$.

\section{Background on Lattice Reduction}
\label{se:background}
We assume the reader is familiar with the geometry of numbers and 
its algorithmic aspects.
Complete introductions to Euclidean lattices algorithmic problems can be found 
in~\cite{MiGo02} and~\cite{Regev04}.

\medskip

\noindent {\bf Gram-Schmidt orthogonalisation.}
Let~$\vec{b}_1,\ldots,\vec{b}_d$ be linearly independent vectors.
Their {\it Gram-Schmidt orthogonalisation}
(GSO)~$\vec{b}_1^*,\ldots,\vec{b}_d^*$ is the orthogonal family
defined recursively as follows: the vector~$\vec{b}_i^*$ is the
component of the vector~$\vec{b}_i$ which is orthogonal to the linear
span of the vectors~$\vec{b}_1,\dots,\vec{b}_{i-1}$.  We
have~$\vec{b}_i^*=\vec{b}_i - \sum_{j=1}^{i-1} \mu_{i,j}\vec{b}_j^* $
where~$\mu_{i,j} = \frac{\ps{\vec{b}_i }{\vec{b}_j^*
}}{\left\|\vec{b}_j^*\right\|^2}$.  For~$i \leq d$ we
let~$\mu_{i,i}=1$.  
Notice that the GSO family depends on the order of the vectors.  If
the~$\vec{b}_i$'s are integer vectors, the~$\vec{b}_i^*$'s and
the~$\mu_{i,j}$'s are rational.

\medskip
\noindent {\bf Lattice volume.}  The volume of a lattice~$L$ is
defined as~$\det(L) = \prod_{i=1}^d \left\|\vec{b}_i^*\right\|$, where
the~$\vec{b}_i$'s are any basis of~$L$. It does not depend on the
choice of the basis of~$L$ and can be interpreted as the
geometric volume of the parallelepiped naturally spanned by 
the basis vectors.

\medskip
\noindent {\bf Minimum and SVP.}  Another important lattice invariant is
the minimum.  The {\it minimum}~$\lambda(L)$ is the radius
of the smallest closed ball centred at the origin containing at least
one non-zero lattice vector. The most famous lattice problem is the
{\it shortest vector problem}. We give here its computational
variant: given a basis of a lattice~$L$, find a lattice vector whose
norm is exactly~$\lambda(L)$.

\medskip
\noindent {\bf CVP.} We give here the computational variant of the
{\em closest vector problem}: given a basis of a lattice~$L$ and a target
vector in the real span of~$L$, find a closest vector of~$L$ 
to the target vector.

\medskip

The volume and the minimum of a lattice cannot behave independently. 
Hermite~\cite{Hermite50} was the first to bound the 
ratio~$\frac{\lambda(L)}{(\det L)^{1/d}}$ as a function of the dimension 
only, but his bound was later on greatly improved by Minkowski in
his {\em Geometrie der Zahlen}~\cite{Minkowski96}.
{\em Hermite's constant}~$\gamma_d$ is defined as the supremum
over~$d$ dimensional lattices~$L$ of the
ratio~$\frac{\lambda(L)^2}{(\det L)^{2/d}}$.  In particular, we
have~$\gamma_d \leq \frac{d+4}{4}$ (see~\cite{Martinet02}), which we
will refer to as~{\em Minkowski's theorem}.  Unfortunately, the proof
of Minkowski's theorem is not constructive. In practice, one often
starts with a lattice basis, and tries to improve its quality. This
process is called lattice reduction. The most usual ones are probably
the LLL and HKZ reductions. Before defining them, we need
the concept of size-reduction.

\medskip
\noindent {\bf Size-reduction.}  A basis~$(\vec{b}_1,\dots,\vec{b}_d)$
is {\em size-reduced} if its GSO family satisfies~$|\mu_{i,j}| \le
1/2$ for all~$1 \le j < i \le d$.  \

\medskip
\noindent {\bf HKZ-reduction.}  A basis~$(\vec{b}_1,\ldots,\vec{b}_d)$
is said to be {\em Hermite-Korkine-Zolotarev-reduced} if it is
size-reduced, the vector~$\vec{b}_1$ reaches the first lattice
minimum, and the projections of the $(\vec{b}_i)_{i\geq 2}$'s
orthogonally to the vector~$\vec{b}_1$ are an HKZ-reduced basis.
The following immediately follows from this definition and
Minkowski's theorem. It is the sole property on HKZ-reduced bases that
we will use:
\begin{lemma}
\label{le:hkz}
If~$(\vec{b}_1,\ldots,\vec{b}_d)$ is HKZ-reduced, then for any~$i\leq d$,
we have:
$$\|\vec{b}_i^*\|\leq 
\sqrt{\frac{d-i+5}{4}} \cdot 
\left( \prod_{j\geq i} \|\vec{b}_j^*\|\right)^{\frac{1}{d-i+1}}.$$
\end{lemma}

HKZ-reduction is very strong, but very expensive to compute. 
On the contrary, LLL-reduction is fairly cheap, but an LLL-reduced basis
is of much lower quality. 

\medskip
\noindent {\bf LLL-reduction~\cite{LeLeLo82}.}
A basis~$(\vec{b}_1,\dots,\vec{b}_d)$ is 
{\em LLL-reduced} if it is size-reduced 
and if its GSO satisfies the~$(d-1)$~Lov\'asz conditions: 
$\frac{3}{4} \cdot \left\| \vec{b}_{\kappa-1}^*
\right\|^2 \le \left\| \vec{b}_{\kappa}^* + \mu_{\kappa,\kappa-1}
\vec{b}_{\kappa-1}^* \right\|^2$.  
The LLL-reduction implies that the norms $\|\vec{b}_{1}^*\|, \dots,
\|\vec{b}_{d}^*\|$ of the GSO vectors never drop too fast:
intuitively, the vectors are not far from being orthogonal. Such bases
have useful properties, like providing exponential approximations to
SVP and CVP. In particular, their first
vector is relatively short.  More precisely:

\begin{theorem}[\cite{LeLeLo82}]
Let~$(\vec{b}_1,\ldots,\vec{b}_d)$ be an LLL-reduced basis 
of a lattice~$L$.
Then we have~$\|\vec{b}_1\| \leq 
2^{\frac{d-1}{4}} \cdot (\det L)^{1/d}$.
Moreover, there exists an algorithm that takes as input any set of integer 
vectors and outputs in deterministic polynomial
time an LLL-reduced basis of the lattice they span.
\end{theorem}

In the following, we will also need the fact that if the set of
vectors given as input to the LLL algorithm starts with a shortest
non-zero lattice vector, then this vector is not changed during the
execution of the algorithm: the output basis starts with the same
vector.

\section{Kannan's SVP Algorithm}
\label{se:kannan}

Kannan's SVP algorithm~\cite{Kannan83} relies on multiple calls to the
so-called short lattice points enumeration procedure. 
The latter aims at computing all vectors of a given
lattice that are in the hyper-sphere centred in~$\vec{0}$ and some
prescribed radius. Variants of the enumeration procedure are described 
in~\cite{AgErVaZe02}.

\subsection{Short Lattice Points Enumeration}

Let $(\vec{b}_1,\ldots,\vec{b}_d)$ be a basis of a lattice~$L \subset
\mZ^n$ and let~$A \in \mZ$.  Our goal is to find all lattice
vectors~$\sum_{i=1}^d x_i \vec{b}_i$ of squared Euclidean 
norm~$\leq A$.
The enumeration works as follows.  Suppose
that~$\left\|\sum_i x_i \vec{b}_i\right\|^2 \leq A$ for some
integers~$x_i$'s. Then, by considering the components of the
vector~$\sum_i x_i \vec{b}_i$ on each of the~$\vec{b}_i^*$'s, we
obtain:
\begin{eqnarray*}
\left(x_d\right)^2 \cdot \|\vec{b}_d^*\|^2 & \leq & A, \\
\left(x_{d-1} + \mu_{d,d-1} x_d\right)^2 \cdot \|\vec{b}_{d-1}^*\|^2 &
\leq & A - \left(x_d\right)^2 \cdot \|\vec{b}_d^*\|^2, \\ 
& \ldots & \\ 
\left(x_i + \sum_{j=i+1}^d \mu_{j,i} x_j\right)^2 \cdot
\|\vec{b}_i^*\|^2 & \leq & A - \sum_{j=i+1}^d l_j, \\ 
& \ldots & \\ 
\left(x_1 + \sum_{j=2}^d \mu_{j,i} x_j\right)^2 \cdot
\|\vec{b}_1\|^2 & \leq & A - \sum_{j=2}^d l_j,
\end{eqnarray*}
where~$l_i = (x_i+\sum_{j>i} x_j \mu_{j, i} )^2 \cdot \|\vec{b}_i^*\|^2$.
The algorithm of Figure~\ref{fig:Enum} mimics the equations above.
It is easy to see that the bit-cost of this algorithm is
bounded by the number of loop iterations times a polynomial in the bit-size 
of the input.  We will prove that
if the input basis~$(\vec{b}_1,\ldots,\vec{b}_d)$ is sufficiently
reduced and if~$A = \|\vec{b}_1\|^2$, then the number of loop
iterations is~$d^{\frac{d}{2e}+o(d)}$.

\begin{figure}[htbp] 
\algorithme{
{\bf Input: } An integral lattice basis~$(\vec{b}_1,\ldots,\vec{b}_d)$, 
a bound~$A \in \mZ$. \\
{\bf Output: } All vectors in~$L(\vec{b}_1,\ldots,\vec{b}_d)$ that are
of squared norm~$\leq A$.\\
\lign Compute the rational~$\mu_{i,j}$'s and~$\|\vec{b}_i^*\|^2$'s. \\
\lign $\vec{x} \assign \vec{0}, \vec{l} \assign \vec{0}, 
S \assign \emptyset$.
\\
\lign $i \assign 1$. While~$i \leq d$, do \\
\lign \ \ \ $l_i \assign (x_i+\sum_{j>i}x_j\mu_{j,i})^2 \|\vec{b}_i^*\|^2$. \\
\lign \ \ \ If $i=1$ and $\sum_{j=1}^d l_j \leq A$, then
$S \assign S \cup \{\vec{x}\}$, $x_1\assign x_1+1$. \\
\lign \ \ \ If $i \neq 1$ and $\sum_{j \geq i} l_j \leq A$, then\\
\lign \ \ \ \ \ \ $i\assign i-1$,
$x_i \assign \left\lceil -\sum_{j>i}(x_j \mu_{j,i}) 
-\sqrt{\frac{A-\sum_{j>i} l_j}{\|\vec{b}_i^*\|^2}} \right\rceil$.\\
\lign \ \ \ If $\sum_{j\geq i} l_j > A$, then 
$i \assign i+1$, $x_i \assign x_i+1$.\\
\lign Return~$S$.
}
\vspace*{-.2cm}
\caption{The Enumeration Algorithm.}
\label{fig:Enum} 
\vspace*{-.5cm} 
\end{figure}

\subsection{Solving SVP}

To solve SVP, Kannan provides an algorithm that computes HKZ-reduced
bases, see Figure~\ref{fig:Kannan}.  The cost of the enumeration
procedure dominates the overall cost and mostly depends on the quality
(i.e., the slow decrease of the~$\|\vec{b}_i^*\|$'s) of the input
basis. The main idea of Kannan's algorithm is thus to spend a lot of
time pre-computing a basis of excellent quality before calling the
enumeration procedure. More precisely, it pre-computes a basis which satisfies
the following definition:

\begin{definition}[Quasi-HKZ-Reduction]
A basis~$(\vec{b}_1,\ldots, \vec{b}_d)$ is quasi-HKZ-reduced if
it is size-red\-uced, if~$\|\vec{b}_2^*\| \geq \|\vec{b}_1^*\|/2$ and if
once projected orthogonally to~$\vec{b}_1$, the other~$\vec{b}_i$'s
are HKZ-reduced.
\end{definition}

\begin{figure}[htbp]
\vspace*{-.6cm} 
\algorithme{ 
{\bf Input: } An integer lattice basis~$(\vec{b}_1,\ldots,\vec{b}_d)$. \\ 
{\bf Output: } An HKZ-reduced basis of the same lattice. \\ 
\lign \ LLL-reduce the basis~$(\vec{b}_1,\ldots,\vec{b}_d)$.  \\
\lign \ Do \\
\lign \ \ \ \ Compute the
projections~$(\vec{b}_i')_{i \geq 2}$ of the~$\vec{b}_i$'s
orthogonally to~$\vec{b}_1$. \\ 
\lign \ \ \ \ HKZ-reduce
the~$(d-1)$-dimensional basis~$(\vec{b}_2',\ldots,\vec{b}_d')$. \\
\lign \ \ \ \ Extend the obtained~$(\vec{b}_i')_{i \geq 2}$'s into vectors
of~$L$ by adding to them rational \\
multiples of~$\vec{b}_1$, in such a
way that we have~$|\mu_{i,1}|\leq 1/2$ for any~$i>1$. \\ 
\lign \ While~$(\vec{b}_1,\ldots,\vec{b}_d)$ is not quasi-HKZ-reduced.\\ 
\lign \ Call the enumeration procedure to find all lattice vectors of 
length~$\leq \|\vec{b}_1\|$. \\
Let~$\vec{b}_0$ be a shortest non-zero vector among them.\\ 
\lign \ $(\vec{b}_1,\ldots,\vec{b}_d)
\assign \mbox{LLL} (\vec{b}_0,\ldots,\vec{b}_d)$.\\
\lign \ Compute the
projections~$(\vec{b}_i')_{i \geq 2}$'s of the~$\vec{b}_i$'s
orthogonally to the vector~$\vec{b}_1$. \\
\lign HKZ-reduce
the~$(d-1)$-dimensional basis~$(\vec{b}_2',\ldots,\vec{b}_d')$. \\
\lign Extend the obtained~$(\vec{b}_i')_{i \geq 2}$'s into vectors
of~$L$ by adding to them rational \\
multiples of~$\vec{b}_1$, in such a
way that we have~$|\mu_{i,1}|\leq 1/2$ for any~$i>1$.
}
\vspace*{-.2cm}
\caption{Kannan's SVP Algorithm.}
\label{fig:Kannan}
\vspace*{-.5cm}
\end{figure}

Several comments need to be made on the algorithm of
Figure~\ref{fig:Kannan}.  Steps~4 and~10 are recursive calls. Nevertheless,
one should be careful because the~$\vec{b}_i'$'s are rational vectors,
whereas the input of the algorithm must be integral. One must
therefore scale the vectors by a common factor.
Steps~5 and~11 can be performed for example by
expressing the reduced basis vectors as integer linear combinations of
the initial ones, using these coefficients to recover lattice vectors
and subtracting a correct multiple of the vector~$\vec{b}_1$.
In Step~7, it is alway possible to choose such a vector~$\vec{b}_0$, since
this enumeration always provides non-zero solutions
(the vector~$\vec{b}_1$ is a one of them).

\subsection{Cost of Kannan's SVP Solver}

We recall briefly Helfrich's complexity analysis~\cite{Helfrich85} of
Kannan's algorithm 
and explain our complexity improvement. 
Let~$C(d, n, B)$ be the worst-case complexity of the algorithm of 
Figure~\ref{fig:Kannan} when given as input a~$d$-dimensional basis which 
is embedded in~$\mZ^n$ and whose
coefficients are smaller than~$B$ in absolute value.
Kannan~\cite{Kannan83} and Helfrich~\cite{Helfrich85} show the following
properties:
\begin{itemize}
\item It computes an HKZ-reduced basis of the lattice
spanned by the input vectors.
\item All arithmetic operations performed during the execution are of
cost~$\mathcal{P}(d, n, \log B)$. This implies that the cost~$C(d, n,
B)$ can be bounded by~$C(d) \cdot \mathcal{P}(\log B, n)$ for some
function~$C(d)$.
\item The number of iterations of the loop of Steps~2--6 
is bounded by~$O(1) + \log d$.
\item The cost of the call to the enumeration procedure at Step~7 
is bounded by~$\mathcal{P}(\log B, n) \cdot d^{d/2+o(d)}$.
\end{itemize}

{From} these properties and those of the LLL algorithm as recalled in 
the previous section, it is easy to obtain the following equation:
$$ C(d) \leq (O(1) + \log d) (C(d-1)+\mathcal{P}(d)) + \mathcal{P}(d) +
d^{\frac{d}{2}+o(d)}.$$  
One can then derive the bound~$C(d,B,n) \leq \mathcal{P}(\log B,n)
\cdot d^{\frac{d}{2}+o(d)}$.

\medskip
The main result of this paper is to improve this complexity upper bound
to~$\mathcal{P}(\log B, n) \cdot d^{\frac{d}{2e}+o(d)}$. In fact, we show 
the following:

\begin{theorem}
\label{th:main}
Given as inputs a quasi-HKZ-reduced 
basis~$(\vec{b}_1,\ldots,\vec{b}_d)$ and~$A=\|\vec{b}_1\|^2$, the number
of loop iterations during the execution of the enumeration algorithm
as described in Figure~\ref{fig:Enum} 
is bounded by~$\mathcal{P}(\log B) \cdot 2^{O(d)} \cdot d^{\frac{d}{2e}}$,
where~$B=\max_i \|\vec{b}_i\|$.
As a consequence, given a $d$-dimensional basis of $n$-dimensional vectors 
whose entries are integers with absolute values~$\leq B$, one can compute
an HKZ-reduced basis of the lattice they span 
in deterministic
time~$\mathcal{P}(\log B,n) \cdot d^{\frac{d}{2e}+o(d)}$.
\end{theorem}

\section{Complexity of the Enumeration Procedure}
\label{se:ana}

This section is devoted to proving Theorem~\ref{th:main}.

\subsection{{From} the Enumeration Procedure to Integer Points in 
Hyper-ellipsoids}

In this subsection, we do not assume anything on the input 
basis~$(\vec{b}_1,\ldots, \vec{b}_d)$ and on the input bound~$A$.
Up to some polynomial in~$d$ and~$\log B$, the complexity of the enumeration
procedure of Figure~\ref{fig:Enum} is the number of loop iterations.
This number of iterations is itself bounded by:
$$\sum_{i=1}^d \left|\left\{\left(x_i,\ldots,x_d\right)\in\mZ^{d-i+1}, 
\|\sum_{j=i}^d x_j\vec{b}_j^{(i)}\|^2 \leq A\right\}\right|,$$
where~$\vec{b}_j^{(i)}= \vec{b}_j - \sum_{k<i} \mu_{j,k} \vec{b}_k^*$ is
the vector~$\vec{b}_j$ once projected orthogonally to the linear span
of the vectors~$\vec{b}_1,\ldots,\vec{b}_{i-1}$.
Indeed, the truncated coordinate~$(x_i, \ldots, x_d)$ is either a valid one,
i.e., we 
have~$|\sum_{j=i}^d x_j\vec{b}_j^{(i)}\|^2 \leq A$, or~$(x_i-1, \ldots, x_d)$
is a valid one, or~$(x_{i+1}, \ldots, x_d)$ is a valid one. In fact, 
if~$(x_i, \ldots, x_d)$ is a valid truncated coordinate, only two non-valid 
ones related to that one can possibly be considered during the execution
of the algorithm: $(x_i+1, \ldots, x_d)$ and~$(x_{i-1}, x_i\ldots, x_d)$ for
at most one integer~$x_{i-1}$.

Consider the 
quantity~$\left|\left\{\left(x_i,\ldots,x_d\right)\in\mZ^{d-i+1}, 
\|\sum_{j=i}^d x_j\vec{b}_j^{(i)}\|^2 \leq A\right\}\right|$. By applying
the change of variable~$x_j 
\leftarrow x_j - \left\lfloor\sum_{k>j} \mu_{k,j} x_k\right\rceil$,
we obtain:

\begin{eqnarray*}
\sum_{i\leq d} |\{(x_i,\ldots,x_d)\in\mZ^{d-i+1} & , &  
\|\sum_{j \geq i} x_j\vec{b}_j^{(i)}\|^2 \leq A \}| \\
& \leq & 
\sum_{i \leq d} |\{(x_i,\ldots,x_d)\in\mZ^{d-i+1},
\sum_{j \geq i} 
(x_j+\sum_{k>j} \mu_{k,j} x_k)^2 \cdot \|\vec{b}_j^*\|^2 \leq A\}| \\
& \leq &
\sum_{i \leq d} |\{(x_i,\ldots,x_d)\in\mZ^{d-i+1},
\sum_{j \geq i} (x_j+\{\sum_{k>j} \mu_{k,j} x_k\ \})^2 \cdot 
\|\vec{b}_j^*\|^2 \leq A\}|.
\end{eqnarray*}

If~$x$ is an integer and~$\eps\in [-1/2,1/2]$, then we have the relation
$(x+\eps)^2 \geq x^2/4$.
If~$x=0$, this is obvious, and otherwise we use the 
inequality~$|\eps|\leq 1/2\leq |x|/2$. 
As a consequence, up to a polynomial factor, the complexity of the enumeration
is bounded by:
$$
\sum_{i \leq d} \left|\left\{\left(x_i,\ldots,x_d\right)\in\
\mZ^{d-i+ 1}, 
\sum_{j \geq i} x_j^2 \cdot \|\vec{b}_j^*\|^2 \leq 4A\right\}\right|.
$$

For any~$i \leq d$, we define the 
ellipsoid~$\E_i = \left\{ (y_i,\ldots,y_d) \in \R^{d-i+1},  
\sum_{j \geq i} y_j^2 \cdot \|\vec{b}_j^*\|^2 \leq 4A \right\}$,
as well as the quantity~$N_i = |\E_i \cap \mZ^{d-i+1}|$. We want
to bound the sum of the~$N_i$'s. We now fix some index~$i$. 
The following sequence of relations is inspired 
from~\cite[Lemma 1]{MaOd90}.
\begin{eqnarray*}
N_i & = & \sum_{(x_i,\ldots,x_d) \in \mZ^{d-i+1}} 
{\bf 1}_{\E_i} (x_i,\ldots,x_d)
\leq \exp \left(d \left(1 - 
\sum_{j \geq i} x_j^2 \frac{\|\vec{b}_j^*\|^2}{4A} \right) \right) \\
& \leq & e^d \cdot \prod_{j \geq i} \sum_{x \in \Z}
\exp \left( -x^2 \frac{d \|\vec{b}_i^*\|^2}{4A} \right) 
= e^d \cdot \prod_{j \geq i}
\Theta \left(\frac{d \|\vec{b}_j^*\|^2}{4A}\right),
\end{eqnarray*}
where $\Theta(t) = \sum_{x \in \Z} \exp(-tx^2)$ is defined for~$t > 0$. 
Notice that $\Theta(t) = 1 + 2 \sum_{x \geq 1} \exp(-tx^2) \leq 
1 + 2 \int_0^\infty \exp(-tx^2) dx = 1 + \sqrt{\frac{\pi}{t}}$. Hence
$\Theta(t) \leq \frac{1 + \sqrt{\pi}}{\sqrt{t}}$ for $t\leq 1$ and 
$\Theta(t) \leq 1 + \sqrt{\pi}$ for~$t\geq 1$. 
As a consequence, we have:
\begin{equation}
\label{eq:interval}
N_i \leq (4e(1 + \sqrt{\pi}))^d \cdot \prod_{j \geq i} \max\left(1, 
\frac{\sqrt{A}}{\sqrt{d} \|\vec{b}_i^*\|}\right).
\end{equation}

One thus concludes that the cost of the enumeration procedure is bounded by:
$$ \mathcal{P}(n, \log A, \log B) \cdot 2^{O(d)} \cdot
\max_{I \subset \iinterv{1,d}} \left(
\frac{(\sqrt{A})^{\card{I}}}{(\sqrt{d})^{\card{I}} 
\prod_{i \in I} \|\vec{b}_i^*\|}\right).$$

\subsection{The Case of Quasi-HKZ-Reduced Bases}

We know suppose that~$A = \|\vec{b}_1\|^2$ and that the input 
basis~$(\vec{b}_1,\ldots,\vec{b}_d)$ is quasi-HKZ-reduced.
Our first step is to strengthen the quasi-HKZ-reducedness hypothesis
to an HKZ-reducedness hypothesis. Let~$I \subset \iinterv{1,d}$. 
If~$1 \notin I$, then, because of the quasi-HKZ-reducedness assumption:
$$\frac{\|\vec{b}_1\|^{\card{I}}}{(\sqrt{d})^{\card{I}} 
\prod_{i \in I} \|\vec{b}_i^*\|} 
\leq 2^d \frac{\|\vec{b}_2^*\|^{\card{I}}}{(\sqrt{d})^{\card{I}} 
\prod_{i \in I} \|\vec{b}_i^*\|}.$$
Otherwise if~$1 \in I$, then we have, by removing~$\|\vec{b}_1^*\|$ from 
the product~$\prod_{i \in I-\{1\}} \|\vec{b}_i^*\|$:
$$\frac{\|\vec{b}_1\|^{\card{I}}}{(\sqrt{d})^{\card{I}} 
\prod_{i \in I} \|\vec{b}_i^*\|} 
\leq 2^d \frac{\|\vec{b}_2^*\|^{\card{I}-1}}{(\sqrt{d})^{\card{I}-1} 
\prod_{i \in I-\{1\}} \|\vec{b}_i^*\|}.$$

As a consequence, in order to obtain Theorem~\ref{th:main}, it suffices 
to prove the following:
\begin{theorem}
\label{th:main4HKZ}
Let $\vec{b}_1,\ldots,\vec{b}_d$ be an HKZ-reduced basis.
Let~$I \subset \iinterv{1,d}$. Then, 
$$\frac{\|\vec{b}_1\|^{\card{I}}}{\prod_{i \in I} \|\vec{b}_i^*\|}
\leq (\sqrt{d})^{\card{I}\left(1 + \log \frac{d}{|I|}\right)} \leq
(\sqrt{d})^{\frac{d}{e} + \card{I}}.$$
\end{theorem}

\subsection{A Property on the Geometry of HKZ-Reduced Bases}
In this section, we prove Theorem~\ref{th:main4HKZ}, which is the last
missing part to obtain the announced result. 
Some parts of the proof are fairly technical and have been 
postponed to the appendix (this is the case for the proofs of 
Lemmata~\ref{le:hermite}--\ref{le:reel}).
As a guide, the reader should consider the typical
case where $(\vec{b}_i)_{1\leq i \leq d}$ is an HKZ-reduced basis for
which~$(\|\vec{b}_i^*\|)_i$ is a non-increasing sequence. In that case,
the shape of the interval~$I$ that is provided by Equation(\ref{eq:interval})
is much simpler: it is an interval~$\iinterv{i,d}$ starting at some index~$i$.
Lemmata~\ref{le:mschnorr} and~\ref{le:hermite} (which should thus be
considered as the core of the proof) and the fact that~$x\log x \geq
-1/e$ for~$x\in [0, 1]$ are sufficient to deal with such simple intervals,
and thus to provide the result.

The difficulties arise when the shape of the set~$I$ under
study becomes more complicated. Though the proof is technically quite
involved, the strategy itself can be summed up in a few words. We
split our HKZ-reduced basis into {\em blocks} (defined by the
expression of~$I$ as a union of intervals), i.e., groups of consecutive
vectors $b_i, b_{i+1}, \dots, b_{j-1}$ such that $i, \ldots, k-1
\not\in I$ and $k, \dots, j-1 \in I$.  The former vectors will be the
``large ones'', and the latter the ``small ones''.
Over each block, Lemma~\ref{le:mschnorr} relates the 
average size of
the small vectors to the average size of the whole block. 
We consider the blocks by decreasing indices (in Lemma~\ref{le:recurs}),
and use an amortised analysis to combine finely the local behaviours
on blocks to obtain a global bound. This recombination is extremely 
tight, and in order to get the desired bound we use 
``parts of vectors'' (non-integral powers of them). This is why we 
need to introduce the~$\pis$ (in Definition~\ref{def:pis}). 
A final convexity argument provided
by Lemma~\ref{le:convexity} gives the result.

In the sequel, $(\vec{b}_i)_{1\leq i \leq d}$ is an HKZ-reduced basis
of a lattice~$L$ of dimension~$d \geq 2$.
\begin{definition}\label{def:pi}
For any~$I \subset \iinterv{1,d}$, we 
define~$\pi_{I} = \left( \prod_{i\in I} \|\vec{b}_i^*\| 
\right)^{\frac{1}{\card{I}}}$. 
Moreover, if~$k \in \iinterv{1, d-1}$, we 
define~$\Gamma_{d}(k) = \prod_{i=d-k}^{d-1} \gamma_{i+1}^{\frac{1}{2i}}$.
\end{definition}

For technical purposes in the proof of Lemma~\ref{le:recurs}, we also need the
following definition.
\begin{definition}\label{def:pis}
If~$1 \leq a < b \leq d$, where~$a$ is
real and~$b$ is an integer, we define:
$$\pis_{[a,b]} = \left( 
\|\vec{b}_{\lfloor a \rfloor}^*\|^{1 - a + \lfloor a \rfloor} \cdot 
\prod_{i=\lfloor a \rfloor + 1}^{b} \|\vec{b}_i^*\|
\right)^{\frac{1}{b+1-a}} 
= \left(\pi_{\iinterv{\lfloor a \rfloor, b}}\right)^
{\frac{(b + 1 -\lfloor a \rfloor)(1-a+\lfloor a \rfloor)}{b+1-a}} 
\cdot 
\left(\pi_{\iinterv{\lfloor a \rfloor + 1, b}}\right)^
{\frac{(b - \lfloor a \rfloor)(a-\lfloor a \rfloor )}{b+1-a}}.
$$
\end{definition}

Note that Definition~\ref{def:pis} naturally 
extends Definition~\ref{def:pi}, 
since~$\pis_{[a, b]} = \pi_{\iinterv{a, b}}$ when~$a$ is an integer.

We need estimates on the order of magnitude of~$\Gamma$, and a technical
lemma allowing us to recombine such estimates. Basically, the following
lemma is a precise version of the identity:
$$\log \Gamma_d(k) \approx \int_{x=d-k}^{d} \frac{x}{2} \log x \, \mathrm{d}x
\approx \frac{\log^2(d) - \log^2(d-k)}{4} \aplt 
\frac{\log d}{2} \log \frac{d}{d-k}.$$

\begin{lemma}
\label{le:hermite}
For all~$1\leq k<d$, we have~$\Gamma_d(k) \leq \sqrt{d}^{\log \frac{d}{d-k}}$.
\end{lemma}


The following lemma derives from the convexity of the function~$x \mapsto 
x \log x$.
\begin{lemma}
\label{le:convexity}
Let $\Delta\geq 1$, and define~$F_\Delta(k, d) = 
\Delta^{-k \log \frac{k}{d}}$.
We have, for all integer~$t$, for all integers~$k_1, \ldots, k_t$ 
and~$d_1, \ldots, d_t$ such that
$1\leq k_i<d_i$ for all~$i\leq t$,
$$\prod_{i \leq t} F_{\Delta}(k_i, d_i)
\leq F_{\Delta}\left(\sum_{i\leq t} k_i, \sum_{i\leq t} d_i\right).$$
\end{lemma}

We now give an ``averaged'' version
of~\cite[Lemma 4]{Schnorr87}. For completeness, we give
its proof in appendix. This provides the result claimed in
Theorem~\ref{th:main4HKZ} for any interval~$I=\iinterv{i,j}$, for
any~$i \leq j \leq d$.
\begin{lemma}
\label{le:mschnorr}
For all~$k \in \iinterv{0, d-1}$, we have
$$\pi_{\iinterv{1,k}} \leq
\left(\Gamma_d(k)\right)^{d/k} \cdot \pi_{\iinterv{k+1, d}} \
\ and \ \
\pi_{\iinterv{k+1,d}} \geq \left(\Gamma_d(k)\right)^{-1} 
\cdot (\det L)^{1/d} \geq \sqrt{d}^{\log\frac{d -k}{d}} (\det L)^{1/d}.$$
\end{lemma}

The following lemma extends Lemma~\ref{le:mschnorr} to the case
where~$k$ is not necessarily an integer. Its proof is conceptually
simple, but involves rather heavy elementary calculus. It would be
simpler to obtain it with a relaxation factor. The result is
nevertheless worth the effort since the shape of the bound is
extremely tractable in the sequel.
\begin{lemma}
\label{le:reel}
If $1 \leq x_1 < x_2 < d$ are real and in~$[1,d)$, 
then~$\pis_{[x_2, d]} \geq \sqrt{d}^{\log \frac{d-x_2}{d-x_1}} \cdot 
\pis_{[x_1, d]}$.
\end{lemma}

We prove Theorem~\ref{th:main4HKZ} by induction on the number of
intervals occurring in the expression of the set~$I$ as a union of
intervals.  The following lemma is the induction step. This is a
recombination step, where we join one block (between the indices~$1$
and~$v$, the ``small vectors'' being those between~$u+1$ and~$v$) to
one or more already considered blocks on its right. An important point
is to ensure that the densities~$\delta_i$ defined below actually
decrease.
\begin{lemma}
\label{le:recurs}
Let~$(\vec{b}_1, \ldots, \vec{b}_d)$ be an HKZ-reduced basis. 
Let~$v\in \iinterv{2, d}$, $I\subset \iinterv{v+1, d}$ 
and~$u \in \iinterv{1, v}$. Assume that:
$$\pi_I^{\card{I}} \geq \prod_{i<t} \left(
\pi_{\iinterv{\alpha_i+1,\alpha_{i+1}}}^{\card{I_i}} \cdot
\sqrt{d}^{\card{I_i} \log \delta_i }\right),$$
where~$I_i = I \cap \iinterv{\alpha_i + 1,\alpha_{i+1}}, 
\delta_i = \frac{\card{I_i}}{\alpha_{i+1}-\alpha_i}$ is the
density of $I$ in~$\iinterv{\alpha_i + 1, \alpha_{i+1}}$, and the \
integers~$t$ and~$\alpha_i$'s, and the densities~$\delta_i$ 
satisfy~$t \geq 1$, 
$v = \alpha_1 < \alpha_2 < \ldots < \alpha_t \leq d$ 
and~$1 \geq \delta_1 > \ldots > \delta_{t-1}>0$.

Then, we have
$$\pi_{I'}^{\card{I'}} \geq \prod_{i<t'} \left(
\pi_{\iinterv{\alpha_i' + 1,\alpha_{i+1}'}}^{\card{I'_i}} 
\cdot
\sqrt{d}^{\card{I'_i} \log \delta_i'} \right),$$
where~$I'=\iinterv{u+1,v} \cup I, 
I'_i = I' \cap \iinterv{\alpha'_i + 1, \alpha'_{i+1}}, 
\delta'_i = \frac{\card{I_i'}}{\alpha_{i+1}'-\alpha_i'}$
and the integers~$t'$ and~$\alpha_i'$'s, and the densities~$\delta'_i$ 
satisfy~$t' \geq 1$, 
$0 = \alpha_1' < \alpha_2' < \ldots < \alpha'_{t'} \leq d$
and~$1 \geq \delta'_1 > \ldots > \delta'_{t'-1} >0$.
\end{lemma}

\begin{proof}
Assume first that $\frac{v-u}{v} \geq \delta_1$, 
Then, thanks to Lemma~\ref{le:mschnorr},
$$
\pi_{I'}^{\card{I'}} = 
\pi_{\iinterv{u+1, v}}^{v-u} \cdot \pi_I^{\card{I}}
\geq 
\pi_{\iinterv{1,v}}^{v-u} \cdot \sqrt{d}^{(v-u)\frac{v-u}{v}} 
\cdot \pi_I^{\card{I}},
$$
we are done with $t'=t+1$, $\alpha'_1 = 1$, $\alpha'_k = \alpha_{k-1}$, 
$\delta'_1 = \frac{v-u}{v}$, $\delta'_k = \delta_{k-1}$. 

Otherwise, we let~$\lu > 0$ be such 
that~$\frac{v-u}{v-\lu} = 
\delta_1 = 
\frac{v-u+ \card{I_1}}{\alpha_2-\lu}$,
where the first equality defines $\lambda_1$ and the second one follows. 
Note that this implies:
$$\pis_{[\lambda_1, v]}^{v-u} 
\cdot \pi_{\iinterv{v+1, \alpha_2}}^{|I_1|}
= \pis_{[\lambda_1, \alpha_2]}^{v-u+|I_1|}.$$

Then, we have, by using Lemma~\ref{le:reel}, 
\begin{eqnarray*}
\pi_{I'}^{\card{I'}} & = & 
\pi_{\iinterv{u+1, v}}^{v-u} \cdot \pi_I^{\card{I}} \\
& \geq &
\left( \pis_{[\lu, v]}^{v-u} 
\cdot \sqrt{d}^{(v-u)\log \frac{v-u}{v-\lu}} \right) 
\cdot \prod_{i<t} \left(
\pi_{\iinterv{\alpha_i + 1,\alpha_{i+1}}}^{\card{I_i}} 
\cdot \sqrt{d}^{\card{I_i} \log \delta_i }\right) \\
& \geq &
\left( \pis_{[\lu, v]}^{v-u} 
\cdot \pi_{\iinterv{v + 1,\alpha_2}}^{\card{I_1}}
\cdot \sqrt{d}^{(v-u)\log \frac{v-u}{v-\lu} + 
\card{I_1} \cdot \log \delta_1} \right)
\cdot \prod_{i=2}^{t-1} \left(
\pi_{\iinterv{\alpha_i + 1,\alpha_{i+1}}}^{\card{I_i}} \cdot 
\sqrt{d}^{\card{I_i} \log \delta_i} \right) \\
& \geq &
\left( \pis_{[\lu, \alpha_2]}^{v-u+\card{I_1}}
\cdot \sqrt{d}^{(v-u+\card{I_1})
\log \frac{v-u+\card{I_1}}{\alpha_2-\lu}}\right)
\cdot \prod_{i=2}^{t-1} \left(
\pi_{\iinterv{\alpha_i + 1,\alpha_{i+1}}}^{\card{I_i}}
\cdot \sqrt{d}^{\card{I_i} \log \delta_i} \right), 
\end{eqnarray*}

If $\frac{v-u+\card{I_1}}{\alpha_2} >
\frac{\card{I_2}}{\alpha_3-\alpha_2}$, we conclude as in the first step,
putting $t'=t$, $\alpha'_1 = 1$, $\alpha'_k = \alpha_k$ for $k \geq 2$, 
$\delta'_1 = (v - u + |I_1|)/\alpha_2$, $\delta'_k = \delta_k$ 
for~$k \geq 2$. 
If this is not the case, we let~$\ld$ be such that:
$$\frac{v-u+\card{I_1}}{\alpha_2-\ld} = 
\delta_2 = 
\frac{v-u+ \card{I \cap \iinterv{\alpha_1 + 1,\alpha_3}}}
{\alpha_3-\ld}.$$ 

\noindent
Notice that since~$\delta_1 = \frac{v - u + |I_1|}{\alpha_2 - \lambda_1} 
> \delta_2$, we have~$\lambda_2 < \lambda_1$. 
A similar sequence of inequalities, using Lemma~\ref{le:reel} to relate
$\pis_{[\lambda_1, \alpha_2]}$ to~$\pis_{[\lambda_2, \alpha_2]}$, leads to 
the following lower bound on~$\pi_{I'}^{\card{I'}}$:
$$\left(\pis_{[\ld, \alpha_3]}^
{v-u+\card{I \cap \iinterv{\alpha_1 + 1,\alpha_3}}} 
\cdot \sqrt{d}^{(v-u+\card{I \cap \iinterv{\alpha_1 + 1,\alpha_3}})
\log \frac{v-u+\card{I \cap \iinterv{\alpha_1 + 1,\alpha_3}}}
{\alpha_3-\ld}}\right) 
\cdot \prod_{i=3}^{t-1} \left(
\pi_{\iinterv{\alpha_i + 1,\alpha_{i+1}}}^{\card{I_i}}
\cdot \sqrt{d}^{\card{I_i} \log \delta_i} \right)$$

We can proceed in the same way, 
constructing~$\lambda_2 > \lambda_3 > \ldots$. 
Suppose first that the construction stops at some point. We have:
\begin{eqnarray*}
\pi_{I'}^{\card{I'}} 
& \geq &
\left( \pi_{\iinterv{1, \alpha_{k+1}}}^
{\card{I' \cap \iinterv{1,\alpha_{k+1}}}} 
\cdot \sqrt{d}^{\card{I' \cap \iinterv{1,\alpha_{k+1}}}
\log \frac{\card{I' \cap \iinterv{1,\alpha_{k+1}}}}
{\alpha_{k+1}}} \right)
\cdot \prod_{i=k+1}^{t-1} \left(
\pi_{\iinterv{\alpha_i+1,\alpha_{i+1}}}^{\card{I_i}}
\sqrt{d}^{\card{I_i} \log \delta_i} \right). 
\end{eqnarray*}
We can then conclude, by putting $t'= t-k+1$, $\alpha'_1 = 1,
\alpha'_j = \alpha_{j+k-1}$ for $j > 1$, 
$\delta'_1 = |I'\cap\iinterv{1, \alpha_{k+1}}|/\alpha_{k+1}$, 
$\delta'_j = \delta_{j+k-1}$ for $j > 1$. 

Otherwise, we end up with:
\begin{eqnarray*}
\pi_{I'}^{\card{I'}} 
& \geq &
\pis_{[\lambda_{t-2}, \alpha_{t-1}]}^{\card{I'}} 
\cdot \sqrt{d}^{\card{I'}
\log \frac{\card{I' \cap \iinterv{1,\alpha_{t-1}}}}
{\alpha_{t-1} - \lambda_{t-2}}}, 
\end{eqnarray*}
to which we can apply Lemma~\ref{le:reel} to 
obtain~$\pi_{I'}^{\card{I'}} 
\geq
\pi_{\iinterv{1, \alpha_{t-1}}}^{\card{I'}}
\cdot \sqrt{d}^{\card{I'} 
\log \frac{\card{I' \cap\iinterv{1, \alpha_{t-1}}}}{\alpha_{t-1}}}$,
which is again in the desired form, with $t'=2$, $\alpha'_1 = 1$, 
$\alpha'_{2} = \alpha_{t-1}$, $\delta'_1 = 
\frac{\card{I' \cap \iinterv{1,\alpha_{t-1}}}}{\alpha_{t-1}}$
\qed
\end{proof}

\noindent Theorem~\ref{th:main4HKZ} now follows from successive 
applications of Lemma~\ref{le:recurs}, as follows:

\medskip
\noindent {\bf Proof of Theorem~\ref{th:main4HKZ}.}
Lemma~\ref{le:recurs} gives us, by induction on the size of the considered
set~$I$, that for all $I\subset \iinterv{1, d}$, we have:

$$\pi_I^{\card{I}} \geq \prod_{i<t} \left(
\pi_{\iinterv{\alpha_i+1,\alpha_{i+1}}}^{\card{I_i}} 
\cdot \sqrt{d}^{\card{I_i} \log \delta_i } \right),$$
where~$I_i = I \cap \iinterv{\alpha_i+1,\alpha_{i+1}}$, and the 
integers~$t$ and~$\alpha_i$'s, and the 
densities~$\delta_i= \frac{\card{I_i}}{\alpha_{i+1}-\alpha_i}$
satisfy~$t \geq 1$, $0 = \alpha_1 < \alpha_2 < \ldots < \alpha_t \leq d$ 
and~$1 \geq \delta_1 > \ldots > \delta_{t-1} >0$.
By using Lemma~\ref{le:convexity} with~$\Delta \assign \sqrt{d},
k_i \assign \card{I_i}$ 
and~$d_i \assign \alpha_{i+1} - \alpha_i$,
we immediately obtain:
$$\pi_I^{\card{I}} \geq 
\left( \sqrt{d}^{\card{I} \log \frac{\card{I}}{\alpha_t-\alpha_1}} \right)
\cdot \left( \prod_{i<t} 
\pi_{\iinterv{\alpha_i + 1,\alpha_{i+1}}}^{\card{I_i} }\right).$$

For convenience, we define~$\delta_t=0$.
Because of the definition of the~$\alpha_i$'s, we have:
\begin{eqnarray*}
\prod_{i<t} \pi_{\iinterv{\alpha_i+1,\alpha_{i+1}}}^{\card{I_i}} 
& = & 
\prod_{i<t} 
\left(\pi_{\iinterv{\alpha_i+1,\alpha_{i+1}}}
^{\alpha_{i+1}-\alpha_i}\right)^{\delta_i}
=  
\prod_{i<t}
\prod_{i \leq j <t}
\left(\pi_{\iinterv{\alpha_i+1,\alpha_{i+1}}}^{\alpha_{i+1}-\alpha_i}\right)
^{\delta_{j} - \delta_{j+1}} \\
& = & 
\prod_{j<t} \left( \prod_{i\leq j}
\pi_{\iinterv{\alpha_i+1,\alpha_{i+1}}}^{\alpha_{i+1}-\alpha_i}\right)
^{\delta_{j} - \delta_{j+1}} 
=
\prod_{j<t}
\left(\pi_{\iinterv{1,\alpha_{j+1}}}^{\alpha_{j+1}}\right)
^{\delta_{j} - \delta_{j+1}}.
\end{eqnarray*}

By using~$t-1$ times Minkowski's theorem, we obtain that:
\begin{eqnarray*}
\pi_I^{\card{I}} & \geq & \sqrt{d}^{|I|\log \frac{|I|}{d}}
\cdot (\|\vec{b}_1\| / \sqrt{d})^{
\sum_{j<t} \alpha_{j+1} (\delta_{j} - \delta_{j+1})} \\
& \geq &
\sqrt{d}^{|I|\log \frac{|I|}{d}} \cdot 
(\|\vec{b}_1\| / 
\sqrt{d})^{\sum_{j<t} (\alpha_{j+1} - \alpha_j) \delta_j}\\
& \geq &
\sqrt{d}^{|I|\left(\log \frac{|I|}{d} - 1\right)}
\cdot \|\vec{b}_1\|^{\card{I}}.
\end{eqnarray*}

The final inequality of the theorem is just the fact that $x \mapsto
x\log(d/x)$ is maximal for~$x = d/e$.  
\qed
\medskip

Note that if~$\max I < d$, we can apply the result to the HKZ-reduced
basis~$(b_1, \dots, b_{\max I})$. In the case where~$I=\{i\}$, we recover
the result of~\cite{Schnorr87} that
\begin{equation}\label{lls}
\|b_i^*\| \geq (\sqrt{i})^{-\log i - 1} \cdot \|b_1\|.
\end{equation}
Still, our result is significantly better to what would have been
obtained by combining several relations of the type of
Equation~(\ref{lls}), when~$|I|$ grows large.  For instance, for a
worst case of our analysis where~$I$ is roughly the
interval~$[d(1-1/e), d]$, this strategy would yield a lower bound of
the form~$\|b_1\|^{d/e} \cdot \sqrt{d}^{(d/e)\log d}$, which is worse
than Helfrich's analysis.

\section{CVP and Other Related Problems}
\label{se:cvp}

In this section, we describe what can be obtained by adapting
our technique to the Closest Vector Problem and other problems
related to strong lattice reduction. We only describe the proofs 
at a high level, since they are relatively straightforward.

In CVP, we are given a basis~$(\vec{b}_1, \ldots, \vec{b}_d)$ and a
target vector~$\vec{t}$, and we look for a lattice vector that is
closest to~$\vec{t}$.  The first step of Kannan's CVP algorithm is to
HKZ-reduce the~$\vec{b}_i$'s. Then one adapts the enumeration
algorithm of Figure~\ref{fig:Enum} for CVP. For the sake of
simplicity, we assume that~$\|\vec{b}_1^*\|$ is the largest of
the~$\|\vec{b}_i^*\|$'s (we refer to Kannan's proof~\cite{Kannan83}
for the general case). By using Babai's nearest hyperplane
algorithm~\cite{Babai86}, we see that there is a lattice
vector~$\vec{b}$ at distance less than~$\sqrt{d} \cdot \|\vec{b}_1\|$
of the target vector~$\vec{t}$. As a consequence, if we take~$A = d
\cdot \|\vec{b}_1\|$ in the adaptation of the enumeration procedure,
we are sure to find a solution. The analysis then reduces
(at the level of Equation~(\ref{eq:interval})) to bound the
ratio~$\frac{{\|\vec{b}_1\|}^d}{\prod_{i\leq d} \|\vec{b}_i^*\|}$,
which can be done with Minkowski's theorem.

\begin{theorem}
Given a basis~$(\vec{b}_1,\ldots,\vec{b}_d)$ and a target vector~$\vec{t}$,
all of them in~$\mR^n$ and 
with integer coordinates whose absolute values are smaller than 
some~$B$, one can find all vectors in the lattice spanned by the~$\vec{b}_i$'s 
that are closest to~$\vec{t}$ in deterministic
time~$\mathcal{P}(\log B, n) \cdot d^{d/2+o(d)}$.
\end{theorem}

The best deterministic complexity bound previously known for this
problem was~$\mathcal{P}(\log B, n) \cdot d^{d+o(d)}$
(see~\cite{Helfrich85,Blomer00}).  Our result can also be adapted to
enumerating all vectors of a lattice that are of length below a
prescribed bound, which is in particular useful in the context of
computing lattice theta series.

Another important consequence of our analysis is a significant
worst-case bound improvement of Schnorr's block-based
strategy~\cite{Schnorr87} to compute relatively short vectors. More
precisely, if we take the bounds given in~\cite{GaHoKoNg06} for
the quality of 
Schnorr's semi-$2k$ reduction and for the transference reduction, we
obtain the table of Figure~\ref{fig:bkz}. 
Each entry of the table
gives the upper bound of the quantity~$\frac{\|\vec{b}_1\|}{(\det 
L)^{1/d}}$ which is reachable for a computational effort of~$2^t$,
for~$t$ growing to infinity. To sum up, the multiplicative 
exponent constant is divided 
by~$e \approx 2.7$.
The table upper bounds can be adapted to the
quantity~$\frac{\|\vec{b}_1\|}{\lambda_1(L)}$ by squaring them.

\begin{figure}[htbp]
$$ \begin{array}{|c|c|c|}
\hline
& \mbox{Semi-$2k$ reduction} & \mbox{Transference reduction}\\
\hline 
\mbox{Using Helfrich's complexity bound} & 
\aplt 2^{\frac{\log 2}{2} \frac{d \log^2 t}{t}} \approx 
2^{0.347 \frac{d \log^2 t}{t}} &
\aplt 2^{\frac{1}{4} \frac{d \log^2 t}{t}} \approx 
2^{0.250 \frac{d \log^2 t}{t}} \\
\hline 
\mbox{Using the improved complexity bound} & 
\aplt 2^{\frac{\log 2}{2e} \frac{d \log^2 t}{t}} \approx 
2^{0.128 \frac{d \log^2 t}{t}} &
\aplt 2^{\frac{1}{4e} \frac{d \log^2 t}{t}} \approx 
2^{0.092 \frac{d \log^2 t}{t}} \\
\hline
\end{array} \vspace*{-.2cm} $$
\caption{Worst-case bounds for block-based reduction algorithms.}
\label{fig:bkz} \vspace*{-.4cm}
\end{figure}

Let us finish by mentioning that work under progress seems to show, by
using a technique due to Ajtai~\cite{Ajtai03}, that our analyses are
sharp, in the sense that for all~$\varepsilon > 0$, we
can build HKZ-reduced bases for which the number of steps of Kannan's
algorithm would be of the order of $d^{d(\frac{1}{2e}-\varepsilon)}$.

\bibliographystyle{plain}
\bibliography{these}

\section*{Proof of Lemma~\ref{le:hermite}}
We prove the result by induction on~$k$.For~$k=1$, the bound easily follows 
from~$\gamma_{d} \leq (d+4)/4$.
Suppose now that the result holds for some~$k \in \iinterv{1,d-2}$, 
and that we want to show that it holds for~$k+1$. Notice that we can suppose 
that~$d \geq 3$. 
Define~$G_d(k) = \frac{1}{2} \log d \log \frac{d}{d-k}$. Then for 
any~$\lambda > 0$,
$$G_d(k+\lambda) - G_d(k) = - \frac12 \log d \log \frac{d-k-\lambda}{d-k}
\geq \frac12 \frac{\lambda \log d}{d - k}.$$

Taking $\lambda = 1$, we see 
that~$G_d(k+1) - G_d(k) \geq \frac12 \frac{\log d}{d - k}$.

{From} the upper bound~$\gamma_d \leq (d+4)/4$, we obtain:
$$\log \Gamma_d(k+1) - \log \Gamma_d(k) =
\frac{1}{2} \frac{\log \gamma_{d-k}}{d-k-1} \leq \frac{1}{2}
\frac{\log (d - k + 4)/4}{d-k-1}.$$

Now, since the 
sequence~$\left(\frac{n \log ((n + 4)/4)}{n - 1}\right)_{n\geq 2}$ is 
increasing, we have:
\begin{eqnarray*}
\frac{(d - k) \log((d - k + 4) / 4)}{d-k-1} 
& \leq & \frac{d - 1}{d - 2} \log ((d + 3) / 4)\\
& = & \log d + \frac{(d-1)\log ((d+3)/4) - (d - 2)\log d}{d - 2} \\
& \leq & \log d,
\end{eqnarray*}
since the last term is a decreasing function of $d$, which is negative 
for~$d=3$.
\qed

\section*{Proof of Lemma~\ref{le:convexity}}
We have~$-\log \prod_{i \leq t} \delta^{-k_i\log \frac{k_i}{d_i}} =
(\log \delta)\cdot \sum_{i\leq t} k_i \log \frac{k_i}{d_i}$.
Now, note that the function~$x \mapsto x \log x$ is convex
on~$[0,+\infty)$.  This means that for any~$t \geq 1$, for
any~$a_1,\ldots,a_t>0$, and for any~$\lu,\ldots,\lambda_t \in [0,1]$ 
such that~$\sum_{i\leq t} \lambda_i=1$, we have:
$$\sum_{i\leq t} \lambda_i a_i \log a_i\geq 
\left(\sum_{i\leq t} \lambda_i a_i \right) \log 
\left(\sum_{i\leq t} \lambda_i a_i \right).$$

In particular, for~$\lambda_i \assign \frac{ d_i}{\sum_{i\leq t}  d_i}$ 
and~$a_i \assign \frac{k_i}{d_i}$,
we get (after multiplication by~$\sum_{i\leq t}  d_i$):
$$-\log \prod_{i \leq t} \delta^{-k_i\log \frac{k_i}{d_i}}
\geq
(\log \delta) \cdot \left( \sum_{i\leq t}  k_i \right) \log 
\left( \frac{\sum_{i\leq t}  k_i}{\sum_{i \leq t}  d_i} 
\right),$$
which is 
exactly~$- \log \delta^{-\left( \sum_{i\leq t} k_i\right) \log 
\frac{\sum_{i\leq t} k_i}{\sum_{i\leq t} d_i}}$.
\qed

\section*{Proof of Lemma~\ref{le:mschnorr}.}
\begin{proof}
We start with the first identity. We prove it by induction
on~$k$.  For~$k=1$, this is Minkowski's bound.  Assume it to be true
for a given~$k\leq d-2$.  We are to prove that it holds for~$k+1$
instead of~$k$. By applying Minkowski's bound to the
$(d-k)$-dimensional HKZ-reduced basis~$\vec{b}_{k+1}^*, \ldots,
\vec{b}_d^*$, we have:
\begin{equation}
\label{eq:minkowski}
\|\vec{b}_{k+1}^*\| \leq \sqrt{\gamma_{d-k}}^{\frac{d-k}{d-k-1}} \cdot
\pi_{\iinterv{k+2, d}}.
\end{equation}

We can rewrite our induction hypothesis as
$$ \pi_{\iinterv{1, k+1}}^{\frac{k+1}{k}} \cdot
\|\vec{b}_{k+1}^*\|^{-\frac{1}{k}} \leq \left( \Gamma_d(k)
\right)^{\frac{d}{k}} \cdot \pi_{\iinterv{k+2, d}}^{\frac{d-k-1}{d-k}}
\cdot \|\vec{b}_{k+1}^*\|^{\frac{1}{d-k}},
$$
or, again, as
$$ \pi_{\iinterv{1, k+1}}^{\frac{k+1}{k}} \leq
\left( \Gamma_d(k) \right)^{\frac{d}{k}} \cdot
\pi_{\iinterv{k+2, d}}^{\frac{d-k-1}{d-k}} \cdot
\|\vec{b}_{k+1}^*\|^{\frac{d}{k(d-k)}}.
$$
This gives, by using Equation~(\ref{eq:minkowski}):
$$
\pi_{\iinterv{1, k+1}}^{\frac{k+1}{k}}
\leq
\left(\Gamma_d(k)\right)^{\frac{d}{k}} \cdot
\sqrt{\gamma_{d-k}}^{\frac{d}{k(d-k-1)}} \cdot
\pi_{\iinterv{k+2, d}}^{\frac{k+1}{k}}
=
\left(\Gamma_d(k+1)\right)^{\frac{d}{k}} \cdot 
\pi_{\iinterv{k+2, d}}^{(k+1)/k}.
$$
By raising this last identity to the power~$\frac{k}{k+1}$, we get
$$
\pi_{\iinterv{1,k+1}} \leq
\left(\Gamma_d({k+1})\right)^{\frac{d}{k+1}} \cdot \pi_{\iinterv{k+2, d}},
$$
which, by induction, yields the first inequality. 
\medskip

The second inequality follows easily from the first one. Indeed, it
suffices to raise the first one to the power~$k/d$, multiply both
sides by~$\left( \pi_{\iinterv{k+1,d}} \right)^{(d-k)/d}$, and use the
identity~$\det L = \left( \pi_{\iinterv{1,k}} \right)^k \cdot \left(
\pi_{\iinterv{k+1,d}} \right)^{d-k}$.
\end{proof}

\section*{Proof of Lemma~\ref{le:reel}.}

First notice that, as 
a consequence of Lemma~\ref{le:mschnorr}, we have, for $k, l$ integers, 
$1\leq k\leq l < d$, 
\begin{equation}
\label{eq:mschnorr}
\pi_{\iinterv{l+1, d}} \geq 
\Gamma_{d-k}(l-k)^{-1} \cdot \pi_{\iinterv{k+1, d}}.
\end{equation}

\noindent
Recall that:
$$
\pis_{[x_1,d]}
= \left(\pi_{\iinterv{\fxu,d}}\right)^{\lu}
\cdot 
\left(\pi_{\iinterv{\fxu+1,d}}\right)^{1-\lu}
\ \ \mbox{and} \ \ 
\pis_{[x_2,d]}
= \left(\pi_{\iinterv{\fxd,d}}\right)^{\ld}
\cdot 
\left(\pi_{\iinterv{\fxd+1,d}}\right)^{1-\ld},
$$
with~$\lambda_i = \frac{(d-\lfloor x_i \rfloor +1)(1-x_i+\lfloor x_i \rfloor)}
{d-x_i+1}$ for~$i\in \{1,2\}$. Notice that since~$x_1<x_2$, 
either~$\fxu +1 \leq \fxd$, or
$\fxu = \fxd$. In the last case, since the 
function~$x \mapsto (u - x)/(v - x)$ is
decreasing when~$u < v$ and for~$x<u$, 
we must have~$\lambda_2 < \lambda_1$. 

We split the proof in several cases, 
depending on the respective values of~$\lu$ and~$\ld$. 
\bigskip

\noindent {\bf First case:} $\lu \leq \ld$.
In that case, we have $\fxu +1 \leq \fxd$.
We define
$$
G :=  \Gamma_{d-\fxu+1}(\fxd - \fxu)^{\lu}
\cdot \Gamma_{d-\fxu}(\fxd - \fxu - 1)^{\ld - \lu}
\cdot \Gamma_{d-\fxu}(\fxd - \fxu)^{1-\ld}.
$$

By using three times Equation~(\ref{eq:mschnorr}), we get:
\begin{eqnarray*}
\pis_{[x_2,d]} & = & \left(\pi_{\iinterv{\fxd, d}}\right)^{\ld}
\cdot 
\left(\pi_{\iinterv{\fxd + 1,d}}\right)^{1-\ld}\\
& \geq & \left(\pi_{\iinterv{\fxd, d}}\right)^{\lu} \cdot 
\left(\pi_{\iinterv{\fxd, d}}\right)^{\ld-\lu} \cdot
\left(\pi_{\iinterv{\fxd + 1,d}}\right)^{1-\ld}\\
& \geq & G^{-1} \cdot
\left(\pi_{\iinterv{\fxu, d}}\right)^{\lu}
\cdot \left(\pi_{\iinterv{\fxu+1, d}}\right)^{1-\lu}.
\end{eqnarray*}

Now, Lemma~\ref{le:mschnorr} gives that 
$$\frac{\log G}{\log \sqrt{d}} \leq \lu \log \frac{d-\fxu + 1}{d-\fxd + 1}
+ (\ld - \lu) 
\log \frac{d-\fxu}{d-\fxd + 1}
+ (1 - \ld) \log \frac{d - \lfloor x_1\rfloor}{d- \fxd},
$$
which, by concavity of the function~$x\mapsto \log x$, is at most the 
logarithm of 
$$
E(x_1, x_2) := \lu \frac{d-\fxu + 1}{d-\fxd + 1} + 
(\ld - \lu) \frac{d-\fxu}{d-\fxd + 1} + 
(1 - \ld) \frac{d - \lfloor x_1\rfloor}{d- \fxd}$$
To complete the proof of this first case, it suffices to prove that 
$E(x_1, x_2) \leq frac{d - x_1}{d - x_2}$.\
We have
\begin{eqnarray*}\label{Eaprescalcul}
E(x_1, x_2) & = & 
\frac{\lambda_1}{d - \fxd + 1} + \frac{d - \fxu}{d - x_2 + 1}\\
& = & \frac{d - x_1}{d - x_2} + 
\frac{\lambda_1}{d - \fxd + 1} - 
\frac{1-\mbox{frac}(x_1)}{d-x_2 + 1} 
- \frac{x_2 - x_1}{(d - x_2)(d - x_2 + 1)},\\
&\leq & \frac{d - x_1}{d - x_2} + \frac{1}{d - x_2 + 1}
\left( \lu - (1 - \mbox{frac}(x_1)) - 
\frac{x_2 - x_1}{d - x_2}\right)\\
& = & \frac{d - x_1}{d - x_2} + \frac{1}{d - x_2 + 1}
\left( \frac{(1-\mbox{frac}(x_1))\mbox{frac}(x_1)}{d-x_1+1} 
- \frac{x_2 - x_1}{d - x_2} \right)\\
&\leq &  \frac{d - x_1}{d - x_2} + \frac{1}{d - x_2 + 1}
\left( \frac{1-\mbox{frac}(x_1)}{d - x_2} - \frac{x_2 - x_1}{d - x_2}\right),
\end{eqnarray*}
from which the result follows at once, since~$\fxu < \fxd$ implies
that~$x_2 - x_1 = \fxd - \fxu + \mbox{frac}(x_2) - \mbox{frac}(x_1) 
\geq 1 - \mbox{frac}(x_1)$. 
\bigskip

\noindent {\bf Second case:}~$\lu > \ld$.
Similarly, defining
$$H = \Gamma_{d-\lfloor x_1\rfloor +1}(\fxd - \lfloor x_1\rfloor)^{\ld} \cdot  
\Gamma_{d-\lfloor x_1\rfloor + 1}(\fxd - \lfloor x_1\rfloor + 1)^{\lu - \ld} 
\cdot
\Gamma_{d-\lfloor x_1\rfloor}(\fxd - \lfloor x_1\rfloor)^{1 - \lu},$$
we obtain 
\begin{eqnarray*}
\pis_{[x_2,d]} & = & \left(\pi_{\iinterv{\fxd, d}}\right)^{\ld}
\cdot 
\left(\pi_{\iinterv{{\fxd}+1,d}}\right)^{1-\ld}\\
& = & \left(\pi_{\iinterv{\fxd, d}}\right)^{\ld}
\left(\pi_{\iinterv{\fxd+1, d}}\right)^{\lu - \ld}
\left(\pi_{\iinterv{\fxd+1, d}}\right)^{1 - \lu}
\\
& \geq &  H^{-1} 
\left(pi_{\iinterv{\fxu, d}}\right)^{\lu} 
\left(\pi_{\iinterv{\fxu+1, d}}\right)^{1 - \lu}.
\end{eqnarray*}

Lemma~\ref{le:mschnorr} gives us that: 
$$\frac{\log H}{\log \sqrt{d}} \leq \ld
\log \frac{d-\fxu + 1}{d-\fxd + 1}
+ (\lu - \ld) 
\log \frac{d-\fxu + 1}{d-\fxd}
+ (1 - \lu) \log \frac{d - \lfloor x_1\rfloor}{d- \fxd}.
$$

By concavity of the function~$x\mapsto \log x$, the right hand side is 
at most the logarithm of 
\begin{eqnarray*}
\ld \frac{d-\fxu + 1}{d-\fxd + 1} + 
(\lu - \ld) \frac{d-\fxu+1}{d-\fxd} & + &
(1 - \lu) \frac{d - \lfloor x_1\rfloor}{d- \fxd} \\ 
& = & E(x_1, x_2) + \frac{\lu - \ld}{(d-\fxd)(d-\fxd + 1)}.
\end{eqnarray*}

Hence, we just need to prove that: 
$$E'(x_1, x_2) := E(x_1, x_2) + \frac{(\lu - \ld)}{(d-\fxd)(d-\fxd + 1)} \leq 
\frac{d - x_1}{d - x_2}.$$

Some elementary calculus provides the equalities:
\begin{eqnarray*}
E'(x_1, x_2) & = & \frac{d-x_1}{d-x_2} + \frac{\lu}{d-\fxd}
-\frac{1-\mbox{frac}(x_2)}{(d - \fxd)(d-x_2+1)} 
-\frac{1-\mbox{frac}(x_1)}{d-x_2+1}
- \frac{x_2-x_1}{(d-x_2)(d-x_2+1)}\\
& = &
\frac{d-x_1}{d-x_2} + \frac{\lu}{d-\fxd} - 
\frac{1 - \mbox{frac}(x_1)}{d - x_2} - 
\frac{1-\mbox{frac}(x_2)}{(d - \fxd)(d-x_2+1)}
- \frac{x_2-\fxu-1}{(d-x_2)(d-x_2+1)}
\end{eqnarray*}

\noindent {\bf Second case, first sub-case:}~$\lu>\ld$, $\fxu < \fxd$.
In that case, 
\begin{eqnarray*} 
E'(x_1, x_2) - \frac{d-x_1}{d-x_2}
& \leq & \frac{\lambda_1 - (1 - \mbox{frac}(x_1))}{d - x_2} 
- \frac{1 - \mbox{frac}(x_2)}{(d - \fxd)(d-x_2+1)} 
- \frac{1}{(d-x_2)(d-x_2+1)}\\
&\leq & \frac{1 - \mbox{frac}(x_1)}{(d-x_2)(d-x_1+1)}
- \frac{1}{(d-x_2)(d-x_2+1)}\\
&\leq & 0
\end{eqnarray*}
\bigskip

\noindent {\bf Second case, second sub-case:}~$\lu>\ld$, $\fxu = \fxd$.
In that case, after some rewriting which can be checked with one's favourite
computer algebra system, one finds that:
\begin{eqnarray*}
E'(x_1, x_2) - \frac{d-x_1}{d-x_2}
& = & \frac{1}{(d-\fxu)(d-x_2)}
\left( \frac{(1-\mbox{frac}(x_1))(x_1-x_2)(d-\fxd)}{d-x_1+1}
- \frac{\mbox{frac}(x_2)(\lambda_1-\lambda_2) }{d-\lfloor x_1 \rfloor +1}  
\right)\\
& \leq & 0.
\end{eqnarray*}
\qed 
\end{document}